\documentclass{amsart}

\usepackage{cite}
\usepackage{amsmath,amssymb,amsfonts}
\usepackage{faktor}
\usepackage{bm}
\usepackage{algorithmic}
\usepackage{graphicx}
\usepackage{textcomp}
\usepackage{xcolor}
\usepackage{rotating, booktabs}
\usepackage{wallpaper}
\usepackage{amsmath}
\usepackage{amsthm}
\usepackage{amsfonts}
\usepackage{amssymb}
\usepackage{graphicx}
\usepackage{float}
\usepackage{subfigure}
\usepackage{makecell}
\usepackage{hyperref}
\usepackage{indentfirst}
\usepackage{comment}
\usepackage{mathrsfs}
\usepackage{tikz}

\makeatletter
\def\@setthanks{\vspace{-\baselineskip}\def\thanks##1{\@par##1\@addpunct.}\thankses}
\makeatother

\newtheorem{definitionenv}{Definition}
\newtheorem{lemmaenv}[definitionenv]{Lemma}
\newtheorem{theoremenv}[definitionenv]{Theorem}
\newtheorem{corollaryenv}[definitionenv]{Corollary}
\newtheorem{propositionenv}[definitionenv]{Proposition}
\newtheorem{conjectureenv}[definitionenv]{Conjecture}
\newtheorem{problemenv}[definitionenv]{Problem}
\newtheorem{remarkenv}[definitionenv]{Remark}
\newenvironment{remark}{\begin{remarkenv}\rm}{\end{remarkenv}}
\newcommand{\er}{\end{remark}}

\newtheorem{exampleenv}{Example}
\newtheorem{app-lemmaenv}[section]{Lemma}

\newenvironment{definition}{\begin{definitionenv}\rm}{\end{definitionenv}}
\newenvironment{lemma}{\begin{lemmaenv}\rm}{\end{lemmaenv}}

\newenvironment{example}{\begin{exampleenv}\rm}{\end{exampleenv}}

\newenvironment{app-lemma}{\begin{app-lemmaenv}\rm}{\end{app-lemmaenv}}

\theoremstyle{definition}

\begin{document}

\title[]{An Error-Correction Model for Information Transmissions of Social Networks}
\author[]{Daqi (Reinhardt) Fang$^1$$^2$$^3$}\thanks{
\noindent
$^1$Hangzhou Yungu School, Hangzhou, China. \\
$^2$Intellisia Institute, Guangzhou, China. \\
$^3$Center for Complex Decision Analysis, Fudan University, Shanghai, China.\\
Email: Reinhardt114514@outlook.com}
\author[]{Pin-Chieh Tseng$^4$}\thanks{\noindent$^4$Institute of Communications Engineering and the Department of Applied Mathematics, National Yang Ming Chiao Tung University (NYCU), Hsinchu 30010, Taiwan.\\
Email: pichtseng@gmail.com.}

\maketitle

\begin{abstract}
We study the error-correction problem of the communication between two vertices in a social network. By applying the concepts of coding theory into the Social Network Analysis (SNA), we develop the code social network model, which can offer an efficient way to ensure the correctness of the message transmission within the social netwoks. The result of this study could apply in vary of social science studies.

\end{abstract}

\section{Introduction}

\subsection{Overview of Apply Coding Theory and SNA}

The application of coding theory to social network analysis has been a topic of great interest in recent years. Coding theory is a branch of mathematics that deals with the efficient transmission of information over noisy channels. It is used to design codes that can be used to encode and decode messages, allowing for reliable communication even in the presence of noise.

Regarding to SNA, it has used in the various algorithms to analysis networks such as pathfinding algorithms (e.g., Dijkstra’s Algorithm), community detection algorithms (e.g., Girvan-Newman Algorithm), link prediction algorithms (e.g., Katz index), and graph embedding algorithms (e.g., DeepWalk, see: Perozzi, Al-Rfou and Skiena \cite{PAS14}) Pathfinding algorithms (Foead et al. \cite{FGKHG21}) are used to find the shortest path between two nodes while community detection algorithms are used to identify groups of nodes that are more densely connected than others in a network (Lancichinetti and Fortunato \cite{LF09}) Link prediction algorithms (Lü and Zhou \cite{LZ11}) are used to predict future connections between nodes while graph embedding algorithms are used to represent networks in low-dimensional vector spaces for further analysis or visualization purposes. Some others worth-noticing algorithms including the Node2Vec (Grover and June 2017), which is an algorithm for learning low-dimensional representations of nodes in a graph, and it uses random walks to explore the graph structure and learn node embeddings that capture both local and global network properties (For example: Hu et al. \cite{HLLL20}) Graph Convolutional Networks (GCN) algorithm that can be used to learn representations of nodes in a graph, and it has been used for tasks such as node classification, link prediction, and community detection (Zhang et al. \cite{ZTXM19}) GraphSAGE is an inductive representation learning algorithm (Hamilton, Ying and Leskovec \cite{HYL17}) for graphs can be used for tasks such as node classification, link prediction, and community detection. 

SNA has been used to study a variety of topics, such as the spread of disease, the diffusion of innovations, and the formation of social movements. Despite its widespread use in many areas, there is still a lack of research and modeling of information transmission within social networks in SNA.

\subsection{Limitations of Information Transmission of the SNA}

Information transmission is an important process in social networks because it allows for the spread of ideas and knowledge among individuals. This process can be studied through various methods such as surveys, interviews, or experiments. However, these methods are often limited in their ability to capture the complexity of information transmission within social networks. As a result, there is a need for more research into how information flows through social networks and how it affects behavior and outcomes.

The first limitation of SNA is that it relies on data from existing networks and hard to simulation. This means that it hard to simulate or capture information about new connections or changes within an existing connections of social network. Additionally, current studies about SNA do not consider the size and transmitted bound in which information is transmitted or received systematically by mathematical method.

A second limitation of SNA is that it does not capture the content of information transmission. While currently SNA can provide insight into how information spreads within a network, it cannot provide insight into what type of information is being shared or how it is being interpreted by different individuals within the network, or how the nodes within a social network made the error and error-correcting. This can lead to inaccurate conclusions about how information spreads within a network and how it affects behavior within the network. For example, what political scientists consider to be a considerable degree of the failure of Zhao Ziyang's faction during the Tiananmen crisis in 1989 (Fewsmith \cite{Few09}; Ziyang \cite{Ziy09}) appears to us to be, to some extent, information transmission errors within his factional network.

Currently, there has two ways to study information transmission within social networks. First one is through Agent-Based Modeling (ABM), which is a type of computer simulation that uses agents—individuals or entities—to represent different actors in a system. These agents interact with each other according to certain rules that are programmed into the model (For example: Frias-Martinez, Williamson and Frias-Martinez \cite{FWF11}; Lane \cite{Lan18}; Li et at. \cite{LYJWL21}) Another way to study information transmission within social networks is through network analysis tools such as graph theory or Exponential Random Graph Models (ERGMs), it provides insights into the structure of a network by analyzing its nodes (individuals) and edges (relationships between individuals).

However, research about information transmission within social networks is still lacking in terms of mathematical algorithms. Algorithms are the backbone of any meaningful SNA, as they are responsible for organizing and managing the information within the social networks. 

\subsection{Structure}
The paper is organized as follows.

Section 2 is a quick review for some basic concept.

In section 3, we will define what is Coding Social Network (CSN), by explaining background parameters, introducing how to encode the most basic simple social network. Then we will classifying the different types of simple CSN

In section 4, we will explain how to perform error detection and correction for complex social networks (usually, a complex social network is composed of multiple simple social networks). And we will encode a “prefect covering social network and give examples.

In section 5, we will discuss the size of a perfect CSN

Section 6 is an application, our aim in this example is to showcase an application in political science of CSN through an example of Chinese elite politics.

\subsection{Academic Contribution of Our Study}

Applying coding theory to SNA can provide valuable insights into the structure and dynamics of complex social systems of social science studies. 

In the context of SNA, mathematically, coding theory can be used to encode the distance, size, error and error-correcting, and upper \& bottom bound of information transmission. For example, it can be used to investigate the amount of noise in a network by encoding messages using an error-correcting code. This allows for more accurate to see transmission of information between nodes in a social network, as well as could reviewing and reducing the amount of information that needs to be transmitted within social networks. 

In practice, firstly, by establishing a CSN, our study could contribute to the field of political factional and elite politics studies by examining the efficiency of factions and simulating the transmission path within these groups. This approach would allow us to identify key players, their relationships, and communication patterns, which can help us understand how factions operate and how they influence decision-making over processes, by using. Coding theory could to analyze CSN data, so that we can detect patterns of information flow of the information transmissions within a political faction and identify potential bottlenecks or vulnerabilities in the political factions.

Another possible application is CSN could be used in studies about the formation of interest groups and simulate the transmission path of the lobby process. This could help solve questions about the relationships between politicians and interest groups. By analyzing the coded data, researchers can identify patterns and trends in social networks between the politicians and lobbers that may not be immediately apparent, and it can then be used to inform policy decisions, improve communication strategies, and enhance social network structures.

Furthermore, CSN can be used in studies about the dynamic social network of social movements and ethnic conflicts. It can help researchers understand why some social movements lead to a bloody end, such as the Tiananmen movement in 1989, or why a specific ethnic conflict happened. By analyzing the patterns of communication and interaction within these networks, CSN can provide insights into the factors that contribute to their success or failure. This information can be used to develop strategies for preventing or resolving conflicts and promoting peaceful coexistence among different groups.

\section{Preliminaries}
In this section, we give a quick review for some basic concepts of error-correction codes. For more details, we recommend readers the book \cite{MS77}.

Let $\mathcal{P}_{n}$ be the collection of all vectors in $F_{q}^{n}$, where $F_{q}$ is a finite field with $q$ elements. A code $C$ of length $n$ is defined as a subset of $\mathcal{P}_{n}$. The elements of $C$ are called codewords. Moreover, for $a = (a_{1}, \dots, a_{n})$, $b(b_{1}, \dots, b_{n}) \in \mathcal{P}_{n}$, we can define their Hamming distance as follows:
\[
d_{H}(a, b) = \lvert \{i : a_{i} \neq b_{i}\} \rvert.
\]
The minimum distance $d$ of a code $C$ is described as the minimum Hamming distance between codewords in $C$. Moreover, the weight of a codeword $a$ is defined as 
\[
{\rm wt}(a) = d(a, \overline{0}),
\]
where $\overline{0}$ is the zero vector. The minimum distance is high related to the ability of error detection and error correction. A code with minimum distance $d$ can detect $d-1$ errors and correct $\lfloor \frac{d-1}{2} \rfloor$ errors. 

For general $C$, to check if a codeword $a$ contains an error, we need to check whether $a \in C$ or not. It may cost lots of time when $q$ is large. In here, we also introduce a special code structure, which provide a convenient way to deal with this problem, called linear code. A linear code $C$ is defined as a $k$-dimensional subspace of $F_{q}^{n}$. Moreover, a linear code can be generated by an $n \times k$ matrix $M$ such that $C = {\rm Im}(M)$. In this case, $M$ is called the generated matrix of $C$. The encoding process $E$ for message $m$ can be defined as follows:
\[
E(m) = M m.
\]
For the matrix $M$, there exists an $(n-k)\times n$ matrix $N$ with rank $n-k$ satisfying $NM = 0$. $N$ is called the parity check matrix and for $a \in F_{q}^{n}$, $a \in C$ if and only if $N a = 0$. Moreover, we denote $C$ by the parameter $[n, k, d]_{q}$, where $n$ is the length of $C$, $k$ is the dimension of $C$ as a subspace and $d$ is the minimum distance.

\begin{example}
    We define $C = [7, 4, 3]_{2}$ by the generated matrix
\[
    M = 
    \begin{pmatrix}
        1 & 0 & 0 & 0\\
        0 & 1 & 0 & 0\\
        0 & 0 & 1 & 0\\
        0 & 0 & 0 & 1\\
        0 & 1 & 1 & 1\\
        1 & 0 & 1 & 1\\
        1 & 1 & 0 & 1\\
    \end{pmatrix}.
\]
Then, the parity check matrix can be taken as
\[
    N = 
    \begin{pmatrix}
        0 & 0 & 0 & 1 & 1 & 1 & 1\\
        0 & 1 & 1 & 0 & 0 & 1 & 1\\
        1 & 0 & 1 & 0 & 1 & 0 & 1\\
    \end{pmatrix}.
\]
\end{example}

\section{Code social network}

\subsection{From social network to code social network}
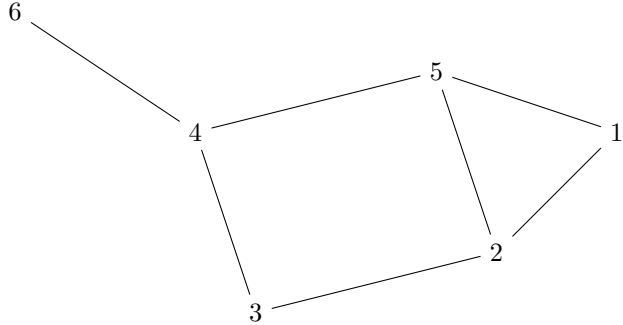
\begin{figure}[b]
    \centering
    \begin{tikzpicture}
    [scale=.8]
    \node (n6) at (1,10) {6};
    \node (n4) at (4,8)  {4};
    \node (n5) at (8,9)  {5};
    \node (n1) at (11,8) {1};
    \node (n2) at (9,6)  {2};
    \node (n3) at (5,5)  {3};

    \foreach \from/\to in   {n6/n4,n4/n5,n5/n1,n1/n2,n2/n5,n2/n3,n3/n4}
    \draw (\from) -- (\to);

    \end{tikzpicture}
    \caption{$S = (V, D)$ with $V = \{1, 2, 3, 4, 5, 6\}$ and $D = \{(1, 2), (1, 5), (2, 3), (2, 5), (3, 4), (4, 5), (4, 6)\}$.}
    \label{fig:ex_S}
\end{figure}

Now, we define $S = (V, D)$ to be a social network with vertex set $V$ and edge set $D$. Note that in the following, we use the word ``vertex" to replace ``node" in social network analysis and regard a social network as a graph. For our description, we give an example in Figure \ref{fig:ex_S}. Then, we define the code social network(CSN) $\hat{S} = (S, C, f)$ as a weighted network which has the same vertex set and edge set as $S$ together with a code $C \subset F_{q}$. Let $a = (a_{1}, \dots, a_{n}) \in C$. An error occurring on $a_{i}$ is assumed as a map sending $a_{i}$ to an element of $F_{q} \setminus a_{i}$ randomly. For $(\alpha, \beta) \in D$, we denote the probability that an error occurs on $a_{i}$ when $\alpha$ sends $a$ to $\beta$ by $p_{(\alpha, \beta)}(a_{i})$. To simplify our model, we assume that the probability $p_{(\alpha, \beta)}(a_{i})$ is independent on the component $i$ and the codeword $a$. Moreover, we denote $p_{(\alpha, \beta)}(a_{i})$ as $p_{(\alpha, \beta)}$. Then, we weight the edge $(\alpha, \beta)$ by the weight function $f$ defined by $f(\alpha, \beta) = p_{(\alpha, \beta)}$. 

\begin{definition}
    Let $\alpha, \beta \in V$ and $P_{(\alpha, \beta)} = (\alpha = \gamma_{0}, \gamma_{1}, \dots, \gamma_{l} = \beta)$ be a path from $\alpha$ to $\beta$ on the weighted network $\hat{S}$. Now, $\alpha$ sends a codeword $a$ along the path ${\rm Path}_{(\alpha, \beta)}$ to $\beta$ and $\beta$ receives the codeword $a'$. Let $E_{(\alpha, \beta)}$ be the expectation value of the Hamming distance $d_{H}(a, a')$. Then, we define the follows.
    \begin{enumerate}
        \item 
        $P_{(\alpha, \beta)}$ is said to be efficient if $ E_{P_{(\alpha, \beta)}} \leq \lfloor \frac{d-1}{2} \rfloor$,
        \item
        $P_{(\alpha, \beta)}$ is said to be semi-efficient if $E_{P_{(\alpha, \beta)}} \leq d-1$,
        \item
        $P_{(\alpha, \beta)}$ is said to be inefficient if $E_{P_{(\alpha, \beta)}} > d-1$.
    \end{enumerate}
\end{definition}

\begin{remark}
    Under our definition and assumption, for two paths $P_{(\alpha, \beta)}, P'_{(\alpha, \beta)}$ with $P_{(\alpha, \beta)}$ has a longer length, we do not have $E_{P(\alpha, \beta)} \geq E_{P'(\alpha, \beta)}$.
\end{remark}

\begin{definition}
    Let $\hat{S} = (S, C, f)$ be a CSN. We define the follows.
    \begin{enumerate}
        \item 
        $\hat{S}$ is said to be efficient if for each $\alpha, \beta \in V$, the shortest path from $\alpha$ to $\beta$ is efficient,
        \item
        $\hat{S}$ is said to be semi-efficient if for each $\alpha, \beta \in V$,  the shortest path from $\alpha$ to $\beta$ is semi-efficient,
        \item
        $\hat{S}$ is said to be inefficient if $\hat{S}$ is not semi-efficient,
        \item
        The critical value $\mathscr{C}(\hat{S})$ of $\hat{S}$ as
        \[
        \mathscr{C}(\hat{S}) = {\rm max}\{l(\alpha, \beta) \mid (\alpha, \beta) \in D\},
        \]
        where $l(\alpha, \beta)$ is denoted by the length of the shortest path from $\alpha$ to $\beta$.
    \end{enumerate}
\end{definition}

\begin{lemma}
    \label{lem:CSN_bin}
    Suppose $f(\alpha, \beta) = p$ is constant for $(\alpha, \beta) \in D$, $C \subset F_{2}^{n}$. Then, 
    \begin{enumerate}
        \item 
        $\hat{S}$ is efficient if and only if
        \[
        n \sum_{k = 0}^{\left\lfloor \frac{l}{2} \right\rfloor} \binom{l}{2k+1} p^{2k+1} (1-p)^{l-2k-1}  \leq \left\lfloor \frac{d-1}{2} \right\rfloor
        \]
        for $1 \leq l \leq \mathscr{C}(\hat{S})$,
        \item
        $\hat{S}$ is semi-efficient if and only if
        \[
        n \sum_{k = 0}^{\left\lfloor \frac{l}{2} \right\rfloor} \binom{l}{2k+1} p^{2k+1} (1-p)^{l-2k-1} \leq d-1
        \]
        for $1 \leq l \leq \mathscr{C}(\hat{S})$,
        \item
        $\hat{S}$ is inefficient if and only if
        \[
        n \sum_{k = 0}^{\left\lfloor \frac{l}{2} \right\rfloor} \binom{l}{2k+1} p^{2k+1} (1-p)^{l-2k-1} > d-1
        \]
        for $1 \leq l \leq \mathscr{C}(\hat{S})$.
    \end{enumerate}
\end{lemma}

\begin{proof}
Since $C \subset F_{2}^{n}$, for a codeword $a = (a_{1}, \dots, a_{n}) \in C$, two errors occur on $a_{i}$ implies $a_{i}$ is not changed. Thus, we only need to consider that for each component, there are only odd edges which cause error. Then, the expectation value of the number of errors which may change the codeword within path of length $l$ can be calculated as
\[
n \sum_{k = 0}^{\left\lfloor \frac{l}{2} \right\rfloor} \binom{l}{2k+1} p^{2k+1} (1-p)^{l-2k-1}.
\]
\end{proof}

\begin{lemma}
    \label{lem:CSN_nonbin}
    Suppose $f(\alpha, \beta) = p$ is constant for $(\alpha, \beta) \in D$ and $C \subset F_{q}^{n}$. We define the sequences $\{A_{i}\}$, $\{B_{i}\}$ by the recurrence relation
    \[
    A_{1} = p, \quad B_{1} = 1-p,
    \]
    \[
    A_{j} = A_{j-1} p \frac{q-2}{q-1} + B_{j-1} p, \quad B_{j} = A_{j-1} p \frac{1}{q-1} + B_{j-1} (1-p).
    \]
    Then, 
    \begin{enumerate}
        \item 
        $\hat{S}$ is efficient if and only if $nA_{l} \leq \lfloor \frac{d-1}{2} \rfloor$ for $1 \leq l \leq \mathscr{C}(\hat{S})$.
        \item
        $\hat{S}$ is semi-efficient if and only if $nA_{l} \leq d-1$ for $1 \leq l \leq \mathscr{C}(\hat{S})$.
    \end{enumerate}
\end{lemma}

\begin{proof}
    According to the above definitions, $A_{i}$ represents the probability that by comparing with the original codeword, an error occurs on one component of the message after passing $i$ edges within a social network. Then, we have the result.
\end{proof}

\subsection{Transmission path on efficient CSN}
\label{sec:trans_path}

Let $\hat{S} = (S, C, f)$ be an efficient CSN. In this part, we give a convenient method to figure out the transmission path from $\alpha$ to $\beta$ by labeling each vertex a codeword in $F_{q}^{N}$ for some $N$ is sufficient large. First, we start with the definition.

\begin{definition}
    Let $\hat{S} = (S, C, f)$ be an efficient CSN. If there is a spanning tree $T = (V, D')$ such that $\hat{T} = (T, C, f')$ with
    \[
    f(\alpha, \beta) = f'(\alpha, \beta) \text{ for } (\alpha, \beta) \in D'
    \]
    is efficient, $\hat{S}$ is called super-efficient.
\end{definition}

For a super-efficient CSN $\hat{S}$ together with an efficient $\hat{T}$, our goal is to label each vertex in $V$ by the function $\phi$, which is an injective map from $V$ to $F_{q}^{N}$ such that for $\alpha, \beta \in V$, the path from $\alpha$ to $\beta$ in $T$ can be determined by comparing $\phi(a)$ with $\phi(b))$. Here, we present the rules for defining $\phi$.
\begin{enumerate}
    \item 
    If $\gamma$ is the root of $T$, $\phi(\gamma)$ is the zero vector.
    \item
    Let $(\alpha, \beta) \in D'$ with $\alpha$ is labeled and $\beta$ is unlabeled. Let $\phi(\alpha) = (a_{1}, \dots, a_{N})$ and $k = {\rm max}\{i : a_{i} \neq 0\}$. If there is $b_{k+1}$ such that the codeword $b = (a_{1}, \dots, a_{k}, b_{k+1}, 0, \dots, 0)$ is not used, let $\phi(\beta) = b$. Otherwise, we label $\beta$ by the vector $(a_{1}, \dots, a_{k}, 0, 1, 0 \dots, 0)$.
\end{enumerate}

Under this labeling, if $\alpha$ wants to send a message to $\beta$ with knowing $\phi(\beta)$, a good transmission path can be found as follows.
\begin{enumerate}
    \item 
    If ${\rm wt}(\phi(\alpha)) \geq {\rm wt}(\phi(\beta))$, $a$ should send the message to the vertex which is labeled as the codeword obtained by changing the last non-zero term in $\phi(\alpha)$ to $0$.
    \item
    If ${\rm wt}(\phi(\alpha)) < {\rm wt}(\phi(\beta))$, $\beta$ should receive the message from the vertex which is labeled as the codeword obtained by changing the last non-zero term in $\phi(\beta)$ to $0$.
\end{enumerate}

\begin{example}
\label{exp:path_tree}
    Let $S = (V, D)$ with $V = \{1, 2, 3, 4, 5, 6\}$ and $D = \{(1, 2)$, $(1, 5)$, $(2, 3), (2, 5), (3, 4), (4, 5), (4, 6)\}$. Let $C \subset F_{2}^{n}$ be a code with minimum distance $d \geq \frac{n}{50} + 1$ and $n = 10$. Let $f(a, b) = p = \frac{1}{100}$ for all $(a, b) \in D$. We define $\hat{S} = (S, C, f)$ and obvious that $\mathscr{C}(\hat{S}) = 3$. Then, by applying Lemma \ref{lem:CSN_bin}, we can check $np = \frac{n}{100}$, $2n p(1-p) = \frac{99n}{5000}$ and
    \begin{align*}
    n \sum_{k = 0}^{\left\lfloor \frac{\mathscr{C}(\hat{S})}{2} \right\rfloor} \binom{\mathscr{C}(\hat{S})}{2k+1} p^{2k+1} (1-p)^{\mathscr{C}(\hat{S})-2k-1} &= n \left(3 p (1-p)^{2} + p^{3}\right)\\
    &= \frac{9802 n}{10^{6}}.
    \end{align*}
    Thus, $\hat{S}$ is efficient. Let $T = (V, D')$ with $D' = \{(6, 4), (4, 3), (4, 5), (5, 2), (5, 1)\}$. We can calculate that $\hat{S}$ is super-efficient together with the spanning tree $T$. We define $\phi$ as a map from $V$ to $F_{2}^{5}$ by 
    \begin{align*}
        &\phi(6) = (0, 0, 0, 0, 0), \quad \phi(4) = (1, 0, 0, 0, 0),\\
        &\phi(3) = (1, 1, 0, 0, 0), \quad \phi(5) = (1, 0, 1, 0, 0),\\
        &\phi(2) = (1, 0, 1, 1, 0), \quad \phi(1) = (1, 0, 1, 0, 1),
    \end{align*}
    and 
    \begin{align*}
    {\rm Im}(\phi) = \{&(0, 0, 0, 0, 0), (1, 0, 0, 0, 0), (1, 1, 0, 0, 0),\\
    &(1, 0, 1, 0, 0), (1, 0, 1, 1, 0), (1, 0, 1, 0, 1)\}.
    \end{align*}
    When vertex $6$ wants to send a message to vertex $1$, ${\rm wt}(\phi(6)) = 0 < 3 = {\rm wt}(\phi(1))$ implies that we should compare $\phi(6)$ with $\phi(1)$ and figure out what $j$ is. By observing $j = 1$, vertex $6$ should send the message to vertex $4$ which is labeled by $(1, 0, 0, 0, 0)$. By repeating this process, vertex $4$ should send the message receiving from vertex $6$ to vertex $5$ with the label $(1, 0, 1, 0, 0)$. Then, vertex $5$ sends the message to vertex $1$.
    \begin{align*}
    &\phi(6) = (0, 0, 0, 0, 0) \rightarrow \phi(4) = (1, 0, 0, 0, 0) \rightarrow \\
    &\phi(5) = (1, 0, 1, 0, 0) \rightarrow \phi(1) = (1, 0, 1, 0, 1).
    \end{align*}
    Now, we consider the case that vertex $3$ wants to send a message to vertex $2$. ${\rm wt}(\phi(3)) = 2 < 3 = {\rm wt}(\phi(2))$ and the $2$th component of $\phi(3)$ is $1$ but $\phi(2)$ contains $0$ in the same position. By our rule, vertex $3$ should send the message to a vertex, which is labeled by $(1, 0, 0, 0, 0)$. That is vertex $4$. Then, vertex $4$ should compare $\phi(4)$ with $\phi(2)$ and figure out the label by changing the third component of $\phi(4)$ to $1$. Then, vertex $4$ should send the message to the vertex $5$. Then, vertex $5$ sends the message to vertex $2$.
    \begin{align*}
    &\phi(3) = (1, 1, 0, 0, 0) \rightarrow \phi(4) = (1, 0, 0, 0, 0) \rightarrow \\
    &\phi(5) = (1, 0, 1, 0, 0) \rightarrow \phi(2) = (1, 0, 1, 1, 0).
    \end{align*}
\end{example}

\begin{remark}
Each vertex in an efficient CSN only needs to keep the code ${\rm Im}(\phi)$ and compare the label of itself with the terminal vertex to find a path, which is good enough in the sense of error-correction (but not necessary to be the shortest path), instead of using the traditional Dijkstra Algorithm to figure out and save every shortest paths.
\end{remark}

\begin{remark}
    For efficient CSNs, we can classify them into the following three classes. 
    \begin{enumerate}
        \item 
        If $S$ is a tree, we say that $\hat{S}$ is in type $A$.
        \item
        If $S$ is dense but not a tree, we say $\hat{S}$ is in type $B$.
        \item
        If $S$ is a complete graph, we say $\hat{S}$ is in type $C$.
    \end{enumerate}
    Readers might refer Figure \ref{fig:type} as an example. Clearly, for type $A$, our method mentioned before always use the shortest path. For type $C$, we can label the vertices by codewords in a simplex code. Under this method, we can identify whether $S$ is a complete graph or not by observing the codewords so that in the case of type $C$, we can use the shortest path to complete our transmission.
\end{remark}

\begin{figure}
    
    \centering
    \begin{tikzpicture}
    [scale=.4]
    \node (n6) at (1,10) {6};
    \node (n4) at (4,8)  {4};
    \node (n5) at (7,10)  {5};
    \node (n1) at (0,8) {1};
    \node (n2) at (7,6)  {2};
    \node (n3) at (1,5)  {3};
    \node (n7) at (4,4)  {7};

    \foreach \from/\to in   {n6/n4,n4/n5,n4/n1,n4/n2,n4/n3,n7/n4}
    \draw (\from) -- (\to);

    \end{tikzpicture}
    \quad \quad
    \begin{tikzpicture}
    [scale=.4]
    \node (n6) at (1,10) {6};
    \node (n4) at (4,8)  {4};
    \node (n5) at (7,10)  {5};
    \node (n1) at (4,11) {1};
    \node (n2) at (7,6)  {2};
    \node (n3) at (1,5)  {3};
    \node (n7) at (4,4)  {7};

    \foreach \from/\to in   {n2/n3,n6/n3,n1/n3,n6/n7,n1/n5,n6/n5,n5/n2,n5/n7,n4/n1,n4/n2,n4/n3,n7/n4}
    \draw (\from) -- (\to);

    \end{tikzpicture}
    \quad \quad
    \begin{tikzpicture}
    [scale=.4]
    \node (n1) at (4,11) {1};
    \node (n2) at (7,6)  {2};
    \node (n3) at (1,5)  {3};
    \node (n4) at (4,4)  {4};

    \foreach \from/\to in   {n1/n2,n1/n3,n1/n4,n2/n3,n2/n4,n3/n4}
    \draw (\from) -- (\to);

    \end{tikzpicture}
    \caption{The graph in the left is type $A$. The middle one is type $B$. The graph in the right is type $C$.}
    \label{fig:type}
\end{figure}
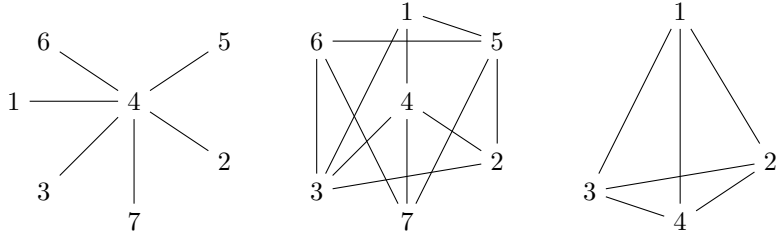

\begin{remark}
    The social network can be seen in many different forms in our society. For example, the political factionalism of the USSR and China is a prime example of a social network. This type of network is characterized by composing of multiple political etlites that are connected through a complex web of alliances and co-interests. Additionally, interest groups lobby-networks in the United States are another example of a non-digraph social network. These networks are composed of various interest groups that are attempting to influence policy decisions through lobbying efforts. Finally, kinship networks within ancient China can also be considered as a type of non-digraph and muti-centralized complex social network (a complex version of our type B in Figure 2) since they involve multiple nodes who are connected through a key kinship node, and this character also made contributions of state building in imperial China (Wang \cite{Wan22}).

\end{remark}

\section{Error detection and correction in Complex CSN}

In this section, we start to apply our model to more complex cases.

\subsection{Covering by super-efficient CSNs}

For a non-efficient CSN $\hat{S}=(S, C, f)$ and vertices $\alpha, \beta$, there might not exist an efficient path from $\alpha$ to $\beta$. To deal with this situation, our approach is to cover each vertices in $\hat{S}$ by super-efficient CSNs $\{\hat{S}_{i}=(S_{i}, C, f)\}$ with each $S_{i} = (V_{i}, D_{i})$ is an induced subgraph of $S$. Furthermore, to allow the transmission between different $\hat{S}_{i}$, we need to use a ``reachable" covering.

\begin{definition}
    Let $N = \{\hat{S}_{i}\}$ be a collection of CSNs.
    \begin{enumerate}
        \item 
        $N$ is said to be a covering of $\hat{S}$ if each $S_{i}$ is an induced subgraph of $S$ and $\bigcup_{i} V_{i} = V$.
        \item
        A covering $N$ of $\hat{S}$ is said to be reachable if for every subset $A$ of $\mathcal{P} = \{1, \dots, \lvert N \rvert\}$, 
        \[
        \bigcup_{i \in A} V_{i} \cap \bigcup_{i \in \mathcal{P} \setminus A} V_{i} \neq \emptyset.
        \]
    \end{enumerate}
\end{definition}

\begin{definition}
    We say that $N = \{\hat{S}_{i}\}$ is an efficient covering of $\hat{S}$ if 
    \begin{enumerate}
        \item 
        $\hat{S}_{i}$ is super-efficient for each $i$,
        \item
        $N$ is a covering of $\hat{S}$, and
        \item
        $N$ is reachable.
    \end{enumerate}
\end{definition}

Under this definition, every efficient covering offers a transmission path for a pair of vertices in $\hat{S}$. Let $\alpha, \beta$ be two vertices. Then, there exist $i, j$ such that $\alpha$ is a vertex of $\hat{S}_{i}$ and $\beta$ is a vertex of $\hat{S}_{j}$. Because of the reachable property, there exist a sequence of super-efficient CSNs $\{\hat{S}_{a_{k}}\}_{k = 1}^{m}$ for some integer $m$ such that 
\begin{enumerate}
    \item $\hat{S}_{a_{0}} = \hat{S}_{i}$,
    \item $\hat{S}_{a_{m}} = \hat{S}_{j}$,
    \item $V_{a_{k}} \cap V_{a_{k+1}} \neq \emptyset$ for all $k$.
\end{enumerate}
Now, the transmission process from $\alpha$ to $\beta$ follows the steps:
\begin{enumerate}
    \item $\alpha$ sends a message to a vertex $\alpha_{1}$ in $V_{a_{0}} \cap V_{a_{1}}$. Then, $\alpha_{1}$ dose once error-correction process.
    \item
    For each $1 
    \leq k \leq m-1$, $\alpha_{k}$ sends a message to a vertex $\alpha_{k+1}$ in $V_{a_{k}} \cap V_{a_{k+1}}$. Then, $\alpha_{k+1}$ dose once error-correction process.
    \item $\alpha_{m}$ sends the message to $\beta$. Then, $\beta$ dose once error-correction and decoding process to obtain the information.
\end{enumerate}

To simplify our model, we suppose $f$ is a constant map. Using the formulas described in Lemma \ref{lem:CSN_bin} and Lemma \ref{lem:CSN_nonbin}, we can estimate an integer $r$ such that every path with length less or equal to $r$ is efficient. For a vertex $\alpha$ and an integer $k$, we define $B_{k}(\alpha)$ as the induced subgraph with vertex set
\[
\{\beta \in V : d_{G}(\alpha, \beta) \leq k\},
\]
where $d_{G}(\alpha, \beta)$ denotes the length of the shortest path from $\alpha$ to $\beta$. Under our assumption, the CSN $(B_{\lfloor\frac{r}{2}\rfloor}(\alpha), C, f)$ is super-efficient.

\begin{definition}
    A subset $M_{r} = \{\alpha_{i}\} \subset V$ is said to be a covering set of $\hat{S}$ if $\{(B_{r}(\alpha_{i}), C, f)\}$ forms an efficient covering of $\hat{S}$. Moreover, the maximum integer $r$ such that the covering set $M_{r}$ of $\hat{S}$ exists is called the radius of $\hat{S}$.
\end{definition}

\begin{remark}
    Finding the smallest covering set of $\hat{S}$ is high related to the problem of minimum $k$-path vertex cover, which has been studied in \cite{BKKS11}, \cite{LZX16}, and \cite{YSM16}.
\end{remark}

Every covering set represents a transmission method for every two vertices of $\hat{S}$. Moreover, the size of a covering set implies the maximum times of error-correction process should be applied during the transmission.

\subsection{Perfect covering}

In this part, we consider CSNs with a great structure, which is called perfect covering.

\begin{definition}
Let $\hat{S}$ be a CSN.
\begin{enumerate}
    \item 
    An efficient covering $N$ of $\hat{S}$ is said to be perfect if for every $\hat{S}_{i}$, $\hat{S}_{j} \in N$, $V_{i} \cap V_{j} \neq \emptyset$.
    \item 
    $\hat{S}$ is called perfect if there exists a perfect covering.
\end{enumerate}
\end{definition}
With the given definition, we observe the simple result.
\begin{lemma}
    Every transmission process between two vertices in a perfect CSN only need to do the error-correction process twice.
\end{lemma}

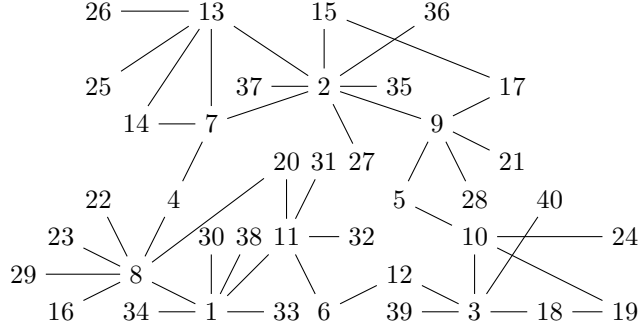
\begin{figure}
    
    \centering
    \begin{tikzpicture}
    [scale=.5]
    \node (n1) at (1,2) {1};
    \node (n2) at (4,8)  {2};
    \node (n3) at (8,2)  {3};
    \node (n4) at (0,5)  {4};
    \node (n7) at (1,7)  {7};
    \node (n8) at (-1,3)  {8};
    \node (n5) at (6,5)  {5};
    \node (n9) at (7,7)  {9};
    \node (n10) at (8,4)  {10};
    \node (n6) at (4,2)  {6};
    \node (n11) at (3,4)  {11};
    \node (n12) at (6,3)  {12};
    \node (n13) at (1,10)  {13};
    \node (n14) at (-1,7)  {14};
    \node (n15) at (4,10)  {15};
    \node (n16) at (-3,2)  {16};
    \node (n17) at (9,8)  {17};
    \node (n18) at (10,2)  {18};
    \node (n19) at (12,2)  {19};
    \node (n20) at (3,6)  {20};
    \node (n21) at (9,6)  {21};
    \node (n22) at (-2,5)  {22};
    \node (n23) at (-3,4)  {23};
    \node (n24) at (12,4)  {24};
    \node (n25) at (-2,8)  {25};
    \node (n26) at (-2,10)  {26};
    \node (n27) at (5,6)  {27};
    \node (n28) at (8,5)  {28};
    \node (n29) at (-4,3)  {29};
    \node (n30) at (1,4)  {30};
    \node (n31) at (4,6)  {31};
    \node (n32) at (5,4)  {32};
    \node (n33) at (3,2)  {33};
    \node (n34) at (-1,2)  {34};
    \node (n35) at (6,8)  {35};
    \node (n36) at (7,10)  {36};
    \node (n37) at (2,8)  {37};
    \node (n38) at (2,4)  {38};
    \node (n39) at (6,2)  {39};
    \node (n40) at (10,5)  {40};

    \foreach \from/\to in   {n1/n8,n8/n4,n4/n7,n7/n2,n2/n9,n9/n5,n5/n10,n10/n3,n1/n11,n11/n6,n6/n12,n12/n3,n13/n7,n14/n7,n15/n2,n16/n8,n17/n9,n18/n3,n19/n18,n10/n19,n20/n11,n9/n21,n22/n8,n20/n8,n15/n17,n13/n14,n23/n8,n24/n10,n13/n2,n25/n13,n26/n13,n2/n27,n9/n28,n29/n8,n30/n1,n11/n31,n11/n32,n1/n33,n1/n34,n2/n35,n2/n36,n2/n37,n1/n38,n3/n39,n3/n40}
    \draw (\from) -- (\to);

    \end{tikzpicture}
    \caption{The graph $S$ constructed by using the vertex set $M = \{1, 2, 3\}$ of a simplex in $\mathbb{R}^{2}$ and $k = 2$.}
    \label{fig:ex3}
\end{figure}

Next, we present a convenient approach to construct a perfect CSN. Let $k$ be given and $M$ be a collection of vertices of a simplex in $\mathbb{R}^{n}$. Now, for every $a, b \in M$, we place a vertex $e$ and two sequences of vertices $\{a_{i}\}_{i = 1}^{k-1}, \{b_{i}\}_{i = 1}^{k-1}$. Then, we add edges $(a, a_{1}), (b, b_{1}), (a_{k-1}, e), (b_{k-1}, e), (a_{i}, a_{i+1}),$ and $(b_{i}, b_{i+1})$ for each $i = 1, \dots k-2$. Moreover, around each $a \in M$, we might add vertices and edges such that all of those vertices can be connected to $a$ by a path with length less than $k$. Under this construction, we obtain a connected graph $S$. By choosing proper $C$ and $f$ such that every path with length less or equal to $2k$ is efficient, we construct a CSN $\hat{S} = (S, C, f)$ with a covering set $M_{k} = M$.

\begin{example}
    Let $M = \{1, 2, 3\}$ be the vertex set of of a simplex in $\mathbb{R}^{2}$ and $k$ = 2. Applying the method we mention before, we construct the CSN $\hat{S}$ by the graph $S$ described in Figure \ref{fig:ex3} and some proper $C$, $f$. Under $\hat{S}$, if vertex $1$ wants to send a message to vertex $17$, vertex $1$ should encode the message by the code $C$ and send the message along an efficient path to vertex $4$ first. Then, vertex $4$ dose once error-correction process. Next, vertex $4$ sends the message to vertex $17$. Final, vertex $17$ implement the error-correction and decoding process to obtain the message.
\end{example}

\section{The size of a perfect CSN}

In this section, we discuss the size of a perfect CSN. First, we give the following definitions.

\begin{definition}
Let $\hat{S}$ be a CSN.
\begin{enumerate}
    \item
    $\hat{S}$ with a covering set $M_{r} = \{\alpha_{i}\}$ is said to have density $e$ if $B_{r}(\alpha_{i})$ contains at most $e$ vertices for each $i$.
    \item 
    If $\hat{S}$ is perfect, the minimum size of a covering set of $\hat{S}$ is called the dimension of $\hat{S}$.
    \item
    Let $\alpha$ be a vertex of $\hat{S}$ with $\alpha \in B_{r}(\alpha_{i})$ and $l^{(n)}(\alpha)$ be the collection of vertices of $V_{i}$ such that for $\beta \in l^{(m)}(\alpha)$, the shortest path from $\alpha$ to $\beta$ has length less or equal to $n$. We define the influence of $\alpha$ with degree $i$ as $\lvert l^{m}(\alpha) \rvert$ and denote it by $L^{(m)}(\alpha)$.
\end{enumerate}
\end{definition}

\begin{remark}
    For $\alpha \in V_{i}$, $\beta \in V_{j}$ with $i \neq j$ and $\alpha, \beta \notin V_{i} \cap V_{j}$, the communication between $\alpha$ and $\beta$ should apply at least once error-correction process before $\beta$ receives the message. It costs time and resource. Thus, we limit the influence of $\alpha$ to the vertices which the transmission process do not need to do any error-correction before the receiver receives the message.
\end{remark}

\begin{lemma}
Let $\alpha$ be a vertex in $B_{r}(\alpha_{i})$. Then, 
    \[
    L^{(m)}(\alpha) \leq 
    \left\{
    \begin{aligned}
        &\lvert V_{i} \rvert, & &\text{ if } \alpha \notin V_{j} \text{ for } j \neq i,\\
        &\sum_{a_{j}}\lvert V_{a_{j}} \rvert, & &\text{ if } \{V_{a_{j}}\}_{j = 1}^{k} \text{ is the maximum set s.t } \alpha \in \bigcap_{a_{j}}V_{a_{j}},
    \end{aligned}
    \right.
    \]
    for each $m$.
\end{lemma}

Now, we give the lower bound and upper bound for the number of vertices of a perfect CSN with dimension $n$ and density $e$.

\begin{lemma}
    Let $\hat{S}$ be a perfect CSN with radius $r$, dimension $n$ and density $e$. Then, 
    \[
    rn + \binom{n}{2} \leq V \leq ne - \binom{n}{2}. 
    \]
\end{lemma}

\begin{proof}
    The upper bound is simple. Every vertex in a covering set represents at most $e$ vertices. However, the vertices in the intersection of two neighborhoods are counted twice. Thus, we obtain the number $ne - \binom{n}{2}$. For the lower bound, each neighborhood of a vertex in a covering set contains at least $r$ vertices outside any intersections. Thus, there are at least $nr + \binom{n}{r}$ vertices.
\end{proof}

\section{Application}

\subsection{Background and Our Aims}

Although the perfect social network is difficult to find in real society ,inefficient social networks which have inefficient information transmission examples are vasely. Inefficient information transmission can lead to information asymmetry, which occurs when one party has access to more or better information than the other, when one party has an informational advantage, they may be able to exploit this advantage to gain a better outcome for themselves at the expense of the other party (For example: \cite{Gun15}, \cite{GWZ18}, \cite{RW17}) The inefficient decision-making, especially in the public policy, is also a significant outcome of inefficient information transmission within the social network (For example: \cite{LM20}, \cite{WM20}, \cite{ZQVEW19})   
  
Our aim in this example is to showcase the inefficiency of Hu Jintao's leadership through an example of Chinese elite politics. According to a multitude of elites, his leadership has been widely criticized for being ineffective. Keller \cite{kel16} has conducted remarkable research on Chinese political elites, and we will utilize her work by coding and applying our algorithm based on the Chinese political elite network of 2012 drawn by Keller. This will enable us to demonstrate the inefficiency of the social network during the Hu Jintao era.

\subsection{Efficient or inefficient}

If Liu Yuan wants to send a message to Hu Chunhua, the shortest path is
\[
\text{Liu Yuan} \rightarrow \text{Xi Jinping} \rightarrow \text{Li Zhanshu} \rightarrow \text{Hu Jintao} \rightarrow \text{Hu Chunhua}.
\]

Take $p = 0.75$, we can calculate
\[
n \sum_{k = 0}^{2} \binom{4}{2k+1} p^{2k+1} (1-p)^{4-2k-1} = 4n \left( p (1-p)^{3} + p^{3} (1-p)\right) = \frac{15n}{32}.
\]
Thus, the transmission path is not efficient for all code with $\frac{15n}{32} \leq \lfloor \frac{d-1}{2} \rfloor$, so isn't the social network. Moreover, there do not exist an efficient covering since the transmission path with length $1$ can not be efficient.

\subsection{Comparative the Path of CCP 2012 Social Network and Efficient Social Network}

If we take $p = 0.001$, we can calculate that for $n = 10$, $d \geq 3$, any path with length less or equal to $100$ is efficient, which implies the social network is efficient. In this case, we only do once error-correction process.

Clearly, for given social network $S$ and a code $C$ with minimum distance greater or equal to $3$ (i.e, $C$ can correct at least $1$ error), there exists a probability $p$ which is small enough such that $\hat{S} = (S, p, C)$ is efficient. However, in some cases, we do not have a communication channel whose $p$ is small enough. For example, when we discuss a communication between several people and each of them sends their message to others using nature language by face to face. There might be lots of misunderstanding. Thus, in such case, $p$ might be very high. Then, our model can offer a convenient method to deal with this situation.

\subsection{Error Detection and Correction of CCP 2012 Social Network}

Now, we consider $p = 0.05$ and $C$ is a binary code with $n = 10$, $d = 3$. By calculation, we know that every path with length less or equal to $2$ is efficient. Thus, the social network we discuss in this section is not efficient. However, it is possible to cover it by efficient networks. For example, we may simply take every important name (i.e the name connects to many people), such as Bo Xilai, Jiang Zemin and so on, as $\alpha_{i}$'s. It offers us a covering $\{B_{1}(\alpha_{i})\}$. Then, the error-correction can be done and the error in the message transmission can be avoided.


\end{document}